%% file: CAI2017_rev2_v11_After_sub.tex
\documentclass{endm}%compile with 3 
\usepackage{endmmacro}
\usepackage{graphicx}
\usepackage{amssymb}
\def\ind{\mathrm{ind}}
\usepackage{lipsum}

\usepackage{tikz-cd}
\usepackage[utf8]{inputenc}
\usepackage[T1]{fontenc}
\usepackage{amsmath,amssymb}
\usepackage{mathptmx}       % selects Times Roman as basic font
\usepackage{helvet}         % selects Helvetica as sans-serif font
\usepackage{courier}        % selects Courier as typewriter font
\usepackage{type1cm}        % activate if the above 3 fonts are
\usepackage{xhfill}                      % not available on your system
\usepackage{xpatch}
\usepackage{textcase}
\makeatletter
\xpatchcmd{\@sect}{\uppercase}{\MakeTextUppercase}{}{}
\xpatchcmd{\@sect}{\uppercase}{\MakeTextUppercase}{}{}
\makeatother
\usepackage{makeidx}         % allows index generation
\usepackage{graphicx}        % standard LaTeX graphics tool
\usepackage{multicol}        % used for the two-column index
\usepackage{pgf,tikz}
\usetikzlibrary{arrows}
\usepackage{ulem}
\def\mod{\mathrm{mod}}

\def\ind{\mathrm{ind}}
\def\dd{\mathbf{d}}
\def\scal#1#2{\langle #1\bv#2 \rangle}
\def\ncp#1#2{#1\langle #2\rangle}
\def\ncs#1#2{#1\langle \!\langle #2\rangle \!\rangle}

\usepackage{textcomp}
\usepackage{enumerate}
\usepackage{tikz}
\usetikzlibrary{arrows,decorations.pathmorphing}
\usetikzlibrary{backgrounds,positioning,fit,petri}
\usetikzlibrary{matrix}
\def\QY{\Q\langle Y\rangle}

\def\scal#1#2{\langle #1\bv#2 \rangle}
\def\bv{\mid}

\def\abs#1{\bv\!#1\!\bv}

\def\ncp#1#2{#1\langle #2\rangle}

\newcounter{per1}
\begingroup
\def\2#1{\ifnum#1<10 0\fi\the#1}
\xdef\isodayandtime{
{\2\day-\2\month-\the\year\space\2{\count0}:%
\2{\count2}}}
\endgroup
\newcommand{\bi}{\begin{itemize}}
\newcommand{\ei}{\end{itemize}}
\newcommand{\bd}{\begin{description}}
\newcommand{\ed}{\end{description}}

\input{macrosB1}
\input{newpack3}

\input{newpack5}

\begin{document}

\begin{verbatim}\end{verbatim}\vspace{2.5cm}

\begin{frontmatter}

\title{Kleene stars of the plane,\\ polylogarithms and symmetries}

\author{ G.H.E. Duchamp\thanksref{coemail}}
\address{LIPN, University Paris Nord, 93800, Villetaneuse, France}

\author{V. Hoang Ngoc Minh\thanksref{coemail1}}
\address{Lille II University, 59024 Lille, France} 

\author{Ngo Quoc Hoan\thanksref{myemail}}
\address{ University of Hai Phong, 171, Phan Dang Luu, Kien An, Hai Phong, Viet Nam} 

	 \thanks[coemail]{Email:
   \href{mailto:gheduchamp@gmail.com} {\texttt{\normalshape
   gheduchamp@gmail.com}}}
	 \thanks[coemail1]{Email:
   \href{mailto:hoang@univ-lille2.fr} {\texttt{\normalshape
   hoang@univ-lille2.fr}}}
   \thanks[myemail]{Email:
   \href{mailto:quochoan_ngo@yahoo.com.vn} {\texttt{\normalshape
   ngoquochoanhp1986@gmail.com}}}

\begin{abstract}
We extend the definition and construct several bases for polylogarithms $\Li_T$,
where $T$ are some series, recognizable by a finite state (multiplicity)
automaton of alphabet\footnote{The space of rational series considered here is
$(\ncp{\C}{X}\shuffle\ncs{\C^\mathrm{rat}}{x_0}\shuffle\ncs{\C^\mathrm{rat}}{x_1}, \shuffle,1_{X^*})$.}
$X=\{x_0,x_1\}$. The kernel of this new ``polylogarithmic map'' $\Li_\bullet$
is also characterized and provides a rewriting process which terminates to a normal form.
We concentrate on algebraic and analytic aspects of this extension allowing 
index polylogarithms at non positive multi-indices, by rational series
and regularize polyzetas at non positive multi-indices\footnote{This research of Ngo Quoc Hoan is funded by Vietnam National Foundation for Science and Technology Development (NAFOSTED) under grant number  101.04-2017.320}.
\end{abstract}

\end{frontmatter}

\section{Introduction}
\quad As a matter of fact, the interest of rational series, over the alphabets
$Y_0=\{y_n\}_{n \in \N}$, $Y=Y_0\setminus\{y_0\}$ and $X=\{x_0,x_1\}$,
is twofold: algebraic and analytic.

Firstly, (from the algebraic point of view) these series are closed
under shuffle products and the shuffle exponential of letters
(and their linear combinations) is precisely their Kleene star\footnote{
{\it i.e.} for any $S\in\ncs{\C}{X}$ such that $\scal{S}{1_{X^*}}=0$,
$S^*$ denotes the sum $1_{X^*}+S+S^2+S^3+\ldots$ and is called its Kleene star (see \cite{berstel}).}.
Secondly, the growth of their coefficients is tame\footnote{
{\it i.e.} for such a rational series $S$ over $X$, there exists
real numbers $K,R>0$ such that, for any $w\in X^*$, the coefficient $\abs{\scal{S}{w}}$
is bounded from above by $K\cdot R^{|w|}$.} \cite{QED,cade,acta}. \cite{QED,cade,acta}
and as such their associated polylogarithms can be rightfully computed
\cite{IMACS,orlando,words03}.

Doing this, we recover many functions (as simple polynomials, for instance)
forgotten in the straight algebra of polylogarithms at positive indices,
which can be viewed as the image of the following isomorphism of algebras \cite{SLC43}
\begin{eqnarray*}
\Li_{\bullet}:(\CX,\shuffle,1_{X^*})&\longrightarrow&(\C\{\Li_w\}_{w\in X^*},\times,1),\\
x_0^{s_1-1}x_1\ldots x_0^{s_r-1}x_1&\longmapsto&\Li_{s_1,\ldots,s_r},\\
\forall n\ge0,x_0^n&\longmapsto&\frac{\log^n(z)}{n!}.
\end{eqnarray*}
To study multi-indexed polylogarithms, one relies on the one-to-one
correspondence between the multi-indices $(-s_1,\ldots,-s_r)$,
in $\Z_{\le0}^r$ (or $(s_1,\ldots,s_r)\in \N_+^r$), and the words $y_{s_1}\ldots y_{s_r}$,
in the monoid $Y_0^*$, indexing polylogarithms by
$y_{s_1}\ldots y_{ s_r}$ as follows \cite{GHMposter,QED,SLC43,FPSAC97}~:
\begin{eqnarray*}
\Li_{y_{s_1}\ldots y_{s_r}}=\Li_{s_1,\ldots,s_r}&\mbox{and}&
\Li^-_{y_{s_1}\ldots y_{s_r}}=\Li_{-s_1,\ldots,-s_r}.
\end{eqnarray*}

We will explain the whole project to extend $\Li_{\bullet}$
over a sub algebra of rational power series.
In particular, we study here various aspects of
$\calC\{\Li_w\}_{w\in X^*}$, where $\calC$ denotes the ring
of polynomials in $z,z^{-1}$ and $(1-z)^{-1}$,
with coefficents in $\C$, and we will express polylogarithms
(resp. harmonic sums) at negative multi-indices as polynomials
in $(1-z)^{-1}$ (resp. $N\in \N$), with coefficients in $\Z$ (resp. $\Q$).

We will concentrate, in particular, on algebraic and analytic aspects
of this extension allowing index polylogarithms 
at non positive multi-indices by rational series and regularize
(divergent) polyzetas at non positive multi-indices.

The paper is structured as follows:
\begin{enumerate}
\item In Section \ref{background}, we will provide some background
consisting in the algebraic combinatorial framework on which the first structures
of polylogarithms and harmonic sums rely, namely their indexation by words and then by (non commutative) polynomials.
\item In Section \ref{integrodiff}, we observe that, as such, the derivations, $z\frac{d}{dz}$ and $(1-z)\frac{d}{dz}$ are continuous (for the standard topology) and so are their sections $\int_{z_0}^z\frac{ds}{s}\bullet$ and $\int_{z_0}^z\frac{ds}{1-s}\bullet$ (for $z_0$ in a suitable non-void open domain) but, in order to satisfy the standard asymptotic condition (see eq. \ref{EDL_uc_asymptotic}), we need other sections which we develop in this paragraph. We will then study the bi-integro-differential
algebra and some functional analysis aspects of polylogarithms, here defined on the
cleft complex plane $\C\setminus(]-\infty,0]\cup[1,+\infty[)$. We will give also
a linear basis of the algebra of polylogarithms with respect to coefficients 
being $\C$-polynomials of $\{z,1/z,1/(1-z)\}$.
\item In Section \ref{shuff_ext}, in order to extend, to classes of rational series, the indexation of $\Li_{\bullet}$, we review properties of rational series and the notation of rational expressions.
\item In Section \ref{trans}, in order to study Kummer functional equations
on polylogarithms, via their non commutative generating series satisfying
noncommutative differential equation, we define polylogarithms on
the universal covering of $\C\setminus\{0,1\}$.
\item Finally, in Section \ref{applications}, the extended double regularization
of divergent polyzetas, on $\CX\shuffle\C[x_0^*]\shuffle\C[(-x_0^*)]\shuffle\C[x_1^*]$
and $\CY\stuffle\C[y_1^*]$, is obtained.
\end{enumerate}
These studies wil be applied to obtain solutions of $KZ_3$
and examples of associators with rational coefficients \cite{CASC2018,inpreparation}.

\section{Polylogarithms and algebraic combinatorial framework}\label{background}

Let us, now, go into details, using the notations of \cite{berstel,reutenauer},

\begin{enumerate}
\item We construct the bialgebras.
\begin{eqnarray*}
(\CX,conc,\Delta_{\shuffle},1_{X^*},\varepsilon)
&\mbox{and}&(\C\langle Y_0\rangle,conc,\Delta_{\stuffle},1_{Y_0^*},\varepsilon)
\end{eqnarray*}
in which, for any $i=0,1$ and $j\geq 0$, one has
\begin{eqnarray*}
\Delta_{\shuffle}(x_i)&=&x_i\otimes1_{X^*}+1_{X^*}\otimes x_i,\\
\Delta_{\stuffle}(y_j)&=&y_j\otimes1_{Y^*}+1_{Y^*}\otimes y_j+\sum_{k+l=j}y_k\otimes y_l.\\
\end{eqnarray*}
and $conc$ is the usual \textit{concatenation product} between noncommutative polynomials. Out of these two, only the first one is Hopf, the last one contains $1+y_0$ which is group-like and has no inverse (see infiltration product phenomenon in \cite{BDHHT}) and therefore has no antipode.

\item Let $\ncs{\C^\mathrm{rat}}{X}$ denote the closure of $\CX$
by rational operations $\{+,.,{}^*\}$ \cite{berstel,IM} (it is closed by shuffle products).
By Kleene-Sch\"utzenberger's theorem (see below paragraph \ref{rats}), any power series $S$ belongs to
$\ncs{\C^\mathrm{rat}}{X}$ if and only if it is {\it recognizable} by an automaton
admitting a {\it linear representation} $(\beta,\mu,\eta)$ of dimension $n$, with
\begin{eqnarray*}
\beta\in{\cal M}_{1,n}(\C),&\mu:X^*\longrightarrow{\cal M}_{n,n}(\C),&\eta\in{\cal M}_{n,1}(\C)
\end{eqnarray*}
such that, for any $w\in X^*$, one has (see \cite{berstel,IM} and paragraph \ref{rats})
\begin{eqnarray*}
\scal{S}{w}=\beta\mu(w)\eta.
\end{eqnarray*}
\item Let us consider the following morphism of algebras, defined by
\begin{eqnarray*}
\pi_X :(\CY,conc,1_{Y^*})&\longrightarrow&(\CX,conc,1_{X^*}),\\
y_{s_1}\ldots y_{s_r}&\longmapsto&x_0^{s_1-1}x_1\ldots x_0^{s_r-1}x_1.
\end{eqnarray*}
It admits an adjoint $\pi_Y$ for the two standard scalar products, {\it i.e.}
\begin{eqnarray*}
\forall p\in\CX,&\forall q\in\CY,&\scal{\pi_Y(p)}{q}_Y=\scal{p}{\pi_X(q)}_X.
\end{eqnarray*}
One checks immediately that $\pi_Y(x_0^{s-1}x_1)=y_s$, $\ker(\pi_Y)=\CX x_0$
and $\pi_Y$ restricted to the subalgebra $(\C\,1_{X^*}\oplus\CX x_1,.)$
is an isomorphism, inverse of $\pi_X$.  
\end{enumerate}

In this work, unless symmetries are involved (i.e. until section \ref{trans}) $\Omega$ denotes the cleft plane ${\mathbb C}\setminus(]-\infty,0]\cup[1,+\infty[)$
and $\calH(\Omega)$, the set of holomorphic functions over the simply connected domain $\Omega$.

The principal object of the present paper, as in \cite{GHMposter,QED},
is the {\it polylogarithm} well defined, for any $(s_1,\ldots,s_r)\in\C^r,r\in\N_+$
and for any $z\in\C$ such that $\abs{z}<1$, by
\begin{eqnarray*}
\Li_{s_1,\ldots,s_r}(z):=\sum_{n_1>\ldots>n_r>0}\frac{z^{n_1}}{n_1^{s_1}\ldots n_r^{s_r}}.
\end{eqnarray*}
So is the following Taylor expansion
\begin{eqnarray*}
\frac{\Li_{s_1,\ldots,s_r}(z)}{1-z}=\sum_{N\ge0}\H_{s_1,\ldots,s_r}(N)\;z^N,
\end{eqnarray*}
where the arithmetic function
\begin{eqnarray*}
\H_{s_1,\ldots,s_r}:\N\longrightarrow\Q
\end{eqnarray*}
is expressed by
\begin{eqnarray*}
\H_{s_1,\ldots,s_r}(N):=\sum_{N\ge n_1>\ldots>n_r>0}\frac1{n_1^{s_1}\ldots n_r^{s_r}}\ .
\end{eqnarray*}
Here, $\Li_{s_1,\ldots,s_r}$ can also be obtained by iterated integrals 
(which provide for free their analytic continuations), along paths in $\Omega$ and with respect to the differential forms
\begin{eqnarray*}
\omega_0(z)=\frac{dz}{z}&\mbox{and}&\omega_1(z)=\frac{dz}{1-z}.
\end{eqnarray*}

Let $\calH_r$ \cite{Goncharov,Zhao} denote the following domain 
\begin{eqnarray*}
\calH_r=\{(s_1,\ldots,s_r)\in\C^r
\vert\forall m =1,\ldots,r;\Re(s_1)+\ldots+\Re(s_m)>m\}\ .
\end{eqnarray*}
After a theorem by Abel, for any $r\ge1$, if $(s_1,\ldots,s_r)\in\calH_r$, we have
\begin{eqnarray*}
\zeta(s_1,\ldots,s_r):=\lim_{z\rightarrow1}\Li_{s_1,\ldots,s_r}(z)
=\lim_{N\rightarrow\infty}\H_{s_1,\ldots,s_r}(N)\ .
\end{eqnarray*}
These limits are no longer valid in the divergent cases and requires the renormalization
of the corresponding divergent polyzetas. This has been done for the case of
polyzetas at positive multi-indices \cite{Daresbury,JSC,cade} and in 
\cite{FKMT,GZ,MP} and completed in \cite{GHMposter,QED} for the case of
negative multi-indices.

The technique used in \cite{GHMposter,QED} (based on encoding polylogarithms,
harmonic sums at positive multi-indices by words in $Y^*$) allows renormalize globally 
polyzetas at non positive multi-indices (via their noncommutative generating series)
but the regularization is not achieved yet. To do that, in the present work,
we introduce the rational series as a new and extended encoding suitable to regularize these functions 
(see \cite{inpreparation} for the analytical justification of such algebraic process).
This technique is already presented in \cite{ISSAC2016}, as a preprint, but never was published before.

\section{Bi-integro-differential algebra of polylogarithms}\label{integrodiff}
\subsection{Differential rings}

Let us consider the following group of transformations of $B=\C\setminus\{0,1\}$ which permutes the singularities in $\{0,1,+\infty\}$
\begin{eqnarray*}
{\cal G}:=\{z\mapsto z,z\mapsto 1-z,z\mapsto z^{-1},z\mapsto(1-z)^{-1},z\mapsto 1-z^{-1},z\mapsto z(z-1)^{-1}\}
\end{eqnarray*}
and let us also consider the following rings~:
\begin{eqnarray}\label{differentialring}
\calC'_0:=\C[z^{-1}],\calC'_1:=\C[(1-z)^{-1}], 
&&\calC_0:=\C[z,z^{-1}],\calC_1:=\C[z,(1-z)^{-1}], \nonumber\\
\calC':=\C[z^{-1},(1-z)^{-1}],
&&\calC:=\C[z,z^{-1},(1-z)^{-1}],
\end{eqnarray}
which are differential rings, endowed with the differential operator
$\partial_z:={d}/{dz}$ and with the neutral element
$1_{\Omega}:\Omega\rightarrow\mathbb{C}$, mapping $z$ to $1_{\Omega}(z)=1$.
It follows that

%===================================================

\begin{lemma}
One has the following properties
\begin{enumerate}
\item For $i=0$ or $1$,
\begin{eqnarray*}
\calC'_i\subsetneq\calC_i\subsetneq\calC&\mbox{and}&\calC'_i\subsetneq\calC'\subsetneq\calC.
\end{eqnarray*}

\item The differential ring $\calC$ is closed under action of ${\cal G}$~:
\begin{eqnarray*}
\forall G(z)\in\calC,&\forall g\in{\cal G},&G(g(z))\in\calC.
\end{eqnarray*}

\item The subrings $\calC_0,\calC_1$ are closed by the involutions
$\{z\mapsto z^{-1},z\mapsto 1-z\}$ and are exchanged
by $\{z\mapsto1-z^{-1},z\mapsto z(z-1)^{-1}\}$, respectively.
\end{enumerate}
\end{lemma}

\begin{proof}
\begin{enumerate}
\item It is immediate from the definitions of $\calC, \calC_i, \calC'_i, i = \{ 0, 1 \}$ in (\ref{differentialring}).
\item This is an easy consequence of the fact that the element $G\in\calC$ can be represented in the form:
\begin{eqnarray*}
G(z)=\sum_{n= 1}^{N_1}\frac{1}{z^n}+\sum_{m=-N_2}^{N_3}(1-z)^m,N_1,N_2,N_3\in\N.
\end{eqnarray*} 
\item It is immediate from the definitions.
\end{enumerate}
\end{proof}

\subsection{Differential and integration operators}

Now, let us consider also the differential operators, acting on $\calH(\Omega)$ \cite{acta}~:
\begin{eqnarray*}
\theta_0=z\frac{d}{dz}&\mbox{and}&\theta_1=(1-z)\frac{d}{dz}
\end{eqnarray*}
and integration operators
\begin{eqnarray*}
\iota_0^{z_0}(f)=\int_{z_0}^z f(s)\omega_0(s)&\mbox{and}&\iota_1^{z_0}(f)=\int_{z_0}^z f(s)\omega_1(s).
\end{eqnarray*}
One has $\theta_i\iota_i^{z_0}=\mathrm{Id}_{\calH(\Omega)}$ (sections of the $\theta_i$).\\
One has other sections of $\theta_i$, defined on $\calC\{\Li_w\}_{w\in X^*}$ named $\iota_i$ (without superscripts).\\ 

They are in fact, much more interesting (and adapted to the explicit computation of associators), these operators ($\iota_i$ without superscripts),
mentioned in the introduction are (more rigorously) defined by means of a $\C$-basis of
\begin{eqnarray*}
\calC\{\Li_w\}_{w\in X^*}=\calC\otimes_\C\C\{\Li_w\}_{w\in X^*}.
\end{eqnarray*}
(The family $(\Li_w)_{w\in X^*}$ can be shown to be free free w.r.t. to all rational functions, see \cite{Linz}.)\\
Now, we recall that a word is \textit{Lyndon} if it is always less (for the lexicographic ordering defined by $x_0<x_1$) than its proper right factors. Their set, noted $\Lyn(X)$, is a transcendance basis of the shuffle algebra $(\ncp{\C}{X},\shuffle,1_{X^*})$ (Radford's theorem \cite{radford}). Then
\begin{eqnarray*}
\C\{\Li_w\}_{w\in X^*}\cong\C[\Lyn(X)],
\end{eqnarray*}
one can partition the alphabet of this polynomial algebra in
\begin{eqnarray*}
(\Lyn(X)\cap X^*x_1)\sqcup \{x_0\}
\end{eqnarray*}
and then get the decomposition 
\begin{eqnarray*}
\calC\{\Li_w\}_{w\in X^*}\simeq
\calC\otimes_\C\C\{\Li_w\}_{w\in X^*x_1}\otimes_\C\C\{\Li_w\}_{w\in x_0^*}.
\end{eqnarray*}
Using the following identity \cite{IMACS},
\begin{eqnarray*}
ux_1x_0^n=ux_1\shuffle x_0^n-\sum_{k=1}^n(u\shuffle x_0^k)x_1x_0^{n-k},
\end{eqnarray*}
we get 
\begin{eqnarray*}
ux_1x_0^n=\sum_{m=0}^n P_mx_1\shuffle x_0^{m},
\end{eqnarray*}
where $P_m\in \ncp{\C}{X}$ is uniquely defined by the above. Thus
\begin{eqnarray*}
\Li_{ux_1x_0^n}(z)=\sum_{m\leq n}\Li_{P_mx_1}(z)\frac{\log^m(z)}{m!}\ .
\end{eqnarray*}
This means that 
\begin{eqnarray*}
\calB&:=&(z^k\Li_{ux_1}(z)\Li_{x_0^n}(z))_{(k,n,u)\in\Z\times\N\times X^*}\\
&\sqcup&((1-z)^{-l}\Li_{ux_1}(z)\Li_{x_0^n}(z))_{(l,n,u)\in\N^+\times\N\times X^*}\\
&\sqcup&(z^k\Li_{x_0^n}(z))_{(k,n)\in \Z\times\N}\\
&\sqcup&((1-z)^{-l}\Li_{x_0^n}(z))_{(l,n)\in \N^+\times\N},
\end{eqnarray*}
is a $\C$-basis of $\calC\{\Li_w\}_{w\in X^*}$.
With this basis, we can define the operator $\iota_0$ as follows 
\begin{definition}\label{def1}
Define the index map $\ind:\calB\rightarrow\Z$ by
\begin{eqnarray*}
\ind (z^k(1-z)^{-l}\Li_{x_0^n}(z))=k&\mbox{and}&\ind(z^k(1-z)^{-l}\Li_{ux_1}(z)\log^n(z))=k+|ux_1|.
\end{eqnarray*}
Then $\iota_0$ is computed as follows
\begin{eqnarray*}
\iota_0(b)=\left\{\begin{array}{rcl}
\displaystyle \int_{0}^z b(s)\omega_0(s),&\mbox{if}&\ind(b)\ge1,\\
\displaystyle \int_{1}^zb(s)\omega_0(s),&\mbox{if}&\ind(b)\le0.
\end{array}\right.
\end{eqnarray*}
and, as $z\not=1$, $\iota_1$ is defined by 
$$
\iota_1(f)=\int_0^zf(s)\omega_1(s)
$$
\end{definition}

We will see in section \ref{discont} that  $\iota_0$ is discontinuous. Nevertheless the pair $\{\iota_0,\iota_1\}$ is adapted to computation of the special solution $\Li_{\bullet}$. One can check easily the following properties.

\begin{proposition}[\cite{GHMposter,QED,orlando,words03}]\label{propthu}
One has the following properties
\begin{enumerate}
\item The operators $\{\theta_0,\theta_1,\iota_0,\iota_1\}$ satisfy in particular,
\begin{eqnarray*}
\theta_1+\theta_0=\bigl[\theta_1,\theta_0\bigr]=\partial_z&\mbox{ and }&\forall k=0,1,\theta_k\iota_k=\mathrm{Id},\\
\left[\theta_0\iota_1,\theta_1\iota_0\right]=0&\mbox{and}&(\theta_0\iota_1)(\theta_1\iota_0)=(\theta_1\iota_0)(\theta_0\iota_1)=\mathrm{Id}.
\end{eqnarray*}

\item The subspace $\calC\{\Li_w\}_{w\in X^*}$ is closed under the action of $\{\theta_0,\theta_1\}$ and $\{\iota_0,\iota_1\}$.
This means that, for any $w=y_{s_1}\ldots y_{s_r}\in Y^*$ (then $\pi_X(w)=x_0^{s_1-1}x_1\ldots x_0^{s_r-1}x_1$)
and $u=y_{t_1}\ldots y_{t_r}\in Y_0^*$, the functions $\Li_w$ and $\Li^-_u$ satisfy
\begin{eqnarray*}
\Li_w=({\iota_0^{s_1-1}\iota_1\ldots\iota_0^{s_r-1}\iota_1})1_{\Omega}
&\mbox{and}&\Li^-_u=({\theta_0^{t_1+1}\iota_1\ldots\theta_0^{t_r+1}\iota_1})1_{\Omega},\\
\iota_0\Li_{\pi_X(w)}=\Li_{x_0\pi_X(w)}&\mbox{and}&\iota_1\Li_{w}=\Li_{x_1\pi_X(w)},\\
\theta_0\Li_{x_0\pi_X(w)}=\Li_{\pi_X(w)}&\mbox{and}&\theta_1\Li_{x_1\pi_X(w)}=\Li_{\pi_X(w)}.
\end{eqnarray*}

\item The bi-integro differential ring $(\calC\{\Li_w\}_{w\in X^*},\theta_0,\iota_0,\theta_1,\iota_1)$ is stable under the action of ${\cal G}$\footnote{When the functions $\Li_w$ and $\calC$ are extended to 
$\widetilde{B}$.}
\begin{eqnarray*}
\forall h\in\calC\{\Li_w\}_{w\in X^*},&\forall g\in{\cal G},&h(g(z))\in\calC\{\Li_w\}_{w\in X^*}.
\end{eqnarray*}

\item $\theta_0\iota_1$ and $\theta_1\iota_0$ are scalar operators within $\calC\{\Li_w\}_{w\in X^*}$,
respectively with eigenvalues $\lambda:=z\rightarrow{z}({1-z})$ and $1/\lambda$, {\it i.e.}
\begin{eqnarray*}
\forall f\in\calC\{\Li_w\}_{w\in X^*},\quad(\theta_0\iota_1)f=\lambda f&\mbox{and}&(\theta_1\iota_0)f=(1/\lambda)f.
\end{eqnarray*}
\end{enumerate}
\end{proposition}
\begin{proof}
The three first points can be checked by (more or less) straightforward computations. The last point needs on the one hand identification of each $\Li_w$ with its unique lifting to $\widetilde{\C\setminus\{0,1\}}$ which coincides with $\Li_w$ on $\Omega$ and on the other hand to lift the elements of ${\cal G}$ so that this group acts on $\widetilde{\C\setminus\{0,1\}}$\footnote{For details, see section \ref{Symm}.}.  
\end{proof}

\subsection{Topology on $\calH(\Omega)$ and continuity or discontinuity of the sections $\iota_i$}\label{discont}

The algebra $\calH(\Omega)$ is that of analytic functions defined over $\Omega$. We quickly describe the standard topology on it (see also \cite{RR}), namely that of {\it compact convergence}
whose seminorms are indexed by compact subsets of $\Omega$, and defined by
\begin{eqnarray*}
p_K(f):=||f||_K=\sup_{s\in K}|f(s)|\ .
\end{eqnarray*}
Of course,
\begin{eqnarray*}
p_{K_1\cup K_2}=\sup(p_{K_1},p_{K_2}),
\end{eqnarray*}
and therefore the same topology is defined by extracting a
{\it fundamental subset of seminorms}, which can be choosen denumerable.
As $\calH(\Omega)$ is complete with this topology it is a Fr\'echet space (see \cite{rudin}) \footnote{It is even a Fr\'echet algebra with unit, but we will not use the multiplicative structure here.}. 

With the standard topology above, an operator $\phi\in\mathrm{End}(\calH(\Omega))$
is continuous iff, with $K_i$ compacts of $\Omega$,
\begin{eqnarray*}
(\forall K_2)(\exists K_1)(\exists M_{21}>0)(\forall f\in \calH(\Omega))(||\phi(f)||_{K_2}\le M_{21}||f||_{K_1}),
\end{eqnarray*}
the algebra $\calC\{\Li_w\}_{w\in X^*}$ (and $\calH(\Omega)$ ) is closed under the operators
$\theta_i,i=0,1$. We have build sections of them $\iota_i^{z_0},\iota_1$, which are continuous and,
$\iota_0$ which is discontinuous and adapted to renormalisation and the computation of associators.      

For $z_0\in \Omega$, let us define $\iota_i^{z_0}\in\mathrm{End}(\calH(\Omega))$ by
\begin{eqnarray*}
\iota_0^{z_0}(f)=\int_{z_0}^z f(s)\omega_0(s)&\mbox{and}&\iota_1^{z_0}(f)=\int_{z_0}^z f(s)\omega_1(s).
\end{eqnarray*}
It is easy to check that $\theta_i\iota_i^{z_0}=\mathrm{Id}_{\calH(\Omega)}$
and that they are continuous on $\calH(\Omega)$ (for the topology of compact convergence) 
because for all $K\subset_\mathrm{compact}\Omega$, we have
\begin{eqnarray*}
|p_K(\iota_i^{z_0}(f)|\leq p_K(f)[\sup_{z\in K}|\int_{z_0}^z\omega_i(s)|],
\end{eqnarray*}
and this is sufficient to prove continuity. The operators $\iota_i^{z_0}$ are also well defined on 
$\calC\{\Li_w\}_{w\in X^*}$ and it is easy to check that 
\begin{eqnarray*}
\iota_i^{z_0}(\calC\{\Li_w\}_{w\in X^*})\subset\calC\{\Li_w\}_{w\in X^*}.
\end{eqnarray*}
Due to the decomposition of $\calH(\Omega)$ into a direct sum of closed subspaces
\begin{eqnarray*}
\calH(\Omega)=\calH_{z_0\mapsto 0}(\Omega)\oplus\C 1_{\Omega},
\end{eqnarray*}
it is not hard to see that the graphs of $\theta_i$ are closed, thus, the  $\theta_i$ are also continuous.

To show discontinuity of $\iota_0$, one of the possibilities consists in exhibiting two sequences $f_n,g_n\in \C\{\Li_w\}_{w\in X^*}$ converging to the same limit but such that
\begin{eqnarray*}
\lim\iota_0(f_n)\neq\lim\iota_0(g_n).
\end{eqnarray*}
Here, we choose the function $z$ for being approached in a twofold way and if $\iota_0$ were continuous, we would have equality of the limits of the image-sequences (and this is not the case). We first remark that  
\begin{eqnarray*}
z=\sum_{n\geq 0}\frac{\log^n(z)}{n!}=\sum_{n\ge1}(-1)^{n+1}\frac{\log^n((1-z)^{-1}}{n!}
\end{eqnarray*}
Set
\begin{eqnarray*}
f_n=\sum_{0\leq m\leq n}\frac{\log^m(z)}{m!}&\mbox{and}&g_n=\sum_{1\leq m\leq n}{(-1)^{m+1}}\frac{\log^m((1-z)^{-1}))}{m!}
\end{eqnarray*}
(these two sequences are in $\C\{\Li_w\}_{w\in X^*}$). It is easily seen that
\begin{eqnarray*}
\iota_0(f_n)=f_{n+1}-1
\end{eqnarray*}
and then 
\begin{eqnarray*}
\lim\limits_{n\to+\infty} \iota_0(f_n)(z) = z -1.
\end{eqnarray*}
Now, for any $s \in[0,z]$ with $z \in ]0,1[$, one has  
\begin{eqnarray*}
|g(s)|&=&|\sum_{m=1}^n(-1)^{m+1}\frac{\log^n(1-s)}{m!}|\\
&\le&\sum_{m=1}^n\frac{|\log^n(1-s)|}{m!}\\
&\le&\exp(-\log(1-s))-1\\
&=&\frac{s}{1-s}.
\end{eqnarray*} 
In order to exchange limits, we apply \textit{Lebesgue's dominated convergence theorem} to the measure space 
$(]0, z],\calB,{dz}/{z})$ ($\calB$ is the usual Borel $\sigma$-algebra) and the function $p(x)=s(1-s)^{-1}$
which is - as are the functions $g_n$ - integrable on $]0, z]$ (for every $z \in ]0,1[$). Then
\begin{eqnarray*}
 \lim(\iota_0(g_n))=\lim\limits_{n\to+\infty}\int\limits_{0}^z g_n(s)\dfrac{ds}{s}
=\int\limits_{0}^z\lim\limits_{n \to+\infty}g_n(s)\dfrac{ds}{s}=\int\limits_0^zs\frac{ds}{s}=z.
\end{eqnarray*}
Hence, for $z\in ]0,1[$, we obtain,  
\begin{eqnarray*}
\lim(\iota_0(f_n))=z-1\ne z=\lim(\iota_0(g_n))\ .
\end{eqnarray*}
This completes the proof.
%\hrule
\section{Extension of $\mathbf{\Li_{\bullet}}$ to its rational domain.}\label{shuff_ext}
\subsection{Rational series}\label{rats} 
Rational series arise from an extension of finite state (boolean) automata to graphs with costs or weights \cite{sam,S}. They have many connections with computer science \cite{berstel,S} but also with operator and Hopf algebras \cite{IM,SwSc}. In short a \textit{weighted graph} is a finite directed graph with edges marked by weights (taken in a semiring, ring or field) and letters (taken in an alphabet $X$) as follows 
$$
<tail>\stackrel{x|\alpha}{\longrightarrow}<head>
$$
this amounts to giving a map $\mu:X\longrightarrow{\cal M}_{n,n}(\C)$ which is extended to words as 
$$
\mu:X^*\longrightarrow{\cal M}_{n,n}(\C)
$$
by morphism. Along a graph path, weights multiply and letters concatenate, this gives the \textit{behaviour} of the automaton, which has an initial vector as input $\beta\in{\cal M}_{1,n}(\C)$ and a final vector $\eta\in{\cal M}_{n,1}(\C)$ as output ; this \textit{behaviour} is a series $S\in \ncs{\C}{X}$.\\
It can be proved the following theorem 
%\hrule
\begin{theorem}[\cite{berstel,IM,S}]
Let $X$ be a finite alphabet and $S\in\ncs{\C}{X}$. The following are equivalent 
\begin{enumerate}[i)]
\item $S$ admits a {\it linear representation} $(\beta,\mu,\eta)$ of dimension $n$ i.e. it exists 
\begin{eqnarray*}
\beta\in{\cal M}_{1,n}(\C),&\mu:X^*\longrightarrow{\cal M}_{n,n}(\C),&\eta\in{\cal M}_{n,1}(\C)
\end{eqnarray*}
such that, for all $w\in X^*$
\begin{equation}
\scal{S}{w}=\beta\mu(w)\eta
\end{equation} 
\item $S$ belongs to the smallest (concatenation) subalgebra of $\ncs{\C}{X}$, containing $\ncp{\C}{X}$ and closed by $S\to S^{-1}$ (rational closure\footnote{Of course $S\to S^{-1}$ is only partially defined, its domain is the set of series such that $\scal{S}{1_{X^*}}\not=0$.} of $\ncp{\C}{X}$).  
\end{enumerate}
\end{theorem}

\subsection{Domain of $\mathbf{\Li_{\bullet}}$}

Under suitable conditions of convergence (see below), the extension of $\Li_{\bullet}$ in general 
and to some subdomain of $\ncs{{\mathbb C}^{\mathrm{rat}}}{X}$ can be done as follows:
call $\mathrm{Dom}(\Li_{\bullet})$ the set of series 
\begin{eqnarray*}
S=\sum_{n\geq 0}S_n&\mbox{with}&S_n:=\sum_{|w|=n}\scal{S}{w}w
\end{eqnarray*}
such that $\sum\limits_{n\geq 0}\Li_{S_n}$ converges uniformly any compact of $\Omega$. Then
\begin{proposition}\label{dom_Li}
One has 
\begin{enumerate}
\item The set $\mathrm{Dom}(\Li_{\bullet})$ is closed by shuffle products.
\item For any $S,T\in\mathrm{Dom}(\Li_{\bullet})$,
one has $\Li_{S\shuffle T}=\Li_S\Li_T$.
\item One has $\CX\shuffle\ncs{\C^{\mathrm{rat}}}{x_0}\shuffle
\ncs{\C^{\mathrm{rat}}}{x_1}\subset\mathrm{Dom}(\Li_{\bullet})$.
\end{enumerate}
\end{proposition}
\begin{proof}
(i) and (ii): Suppose $S,T\in \mathrm{Dom}(\Li_{\bullet})$ and $S=\sum_{p\geq 0}S_p$ (resp. 
$T=\sum_{q\geq 0}T_q$) their decomposition in homogeneous components, then the family $(\Li_{S_p})_{p\geq 0}$
(resp. $(\Li_{T_q})_{q\geq 0}$) is summable in $\calH(\Omega)$. This implies (\cite{BGT} Ch III \S 6) that the families 
$(\Li_{S_p}\Li_{T_q})_{p,q\geq 0}$ and then $(\sum_{p+q=n}\Li_{S_p}\Li_{T_q})_{n\geq 0}$ are summable in 
$\calH(\Omega)$. As  
$$
\sum_{p+q=n} \Li_{S_p}\Li_{T_q}=\sum_{p+q=n}\Li_{S_p\shuffle T_q}=\Li_{(S\shuffle T)_n}
$$
one gets that $S\shuffle T\in \mathrm{Dom}(\Li_{\bullet})$ and $\Li_{S\shuffle T}=\Li_S\Li_T$.\\
(iii): In view of (i,ii) it suffices to check that each of 
$\CX,\ncs{\C^{\mathrm{rat}}}{x_0},\ncs{\C^{\mathrm{rat}}}{x_1}$ is in $\mathrm{Dom}(\Li_{\bullet})$. The first being given, the property for the last two is a consequence of Kronecker's theorem \cite{zyg} i.e. the fact that  
$$
\ncs{\C^{\mathrm{rat}}}{x}=\{P/Q\}_{P,Q\in \C[x]\atop Q(0)\not=0}
$$
and the partial fraction decomposition. 
\end{proof}

This extension is compatible with identities between rational series
as {\it Lazard's elimination}\footnote{{\it i.e.}
$X^*=(x_0^*x_1)^*x_0^*$. In other words, in $\ncs{\C}{X}$, 
$(1-(x_0+x_1))^{-1}=(1-(1-x_0)^{-1}x_1))^{-1}(1-x_0)^{-1}$.}, for instance, for all 
$S\in\ncs{\C^{\mathrm{rat}}}{S}\cap \mathrm{Dom}(\Li_{\bullet})$~:
\begin{eqnarray*}
\Li_{S}(z)=\Sum_{n\ge0}\scal{S}{x_0^n}\Frac{\log^n(z)}{n!}
+\Sum_{k\ge1}\Sum_{w\in(x_0^*x_1)^kx_0^*}\scal{S}{w}\Li_w(z),\label{lazard}
\end{eqnarray*}
\begin{remark}
Here we will be mostly interested by rational series within $\mathrm{Dom}(\Li_{\bullet})$.
But  there are other series as the following (infinite sum of rational series). 
\begin{eqnarray*}
T=\sum_{n\ge0}\frac{(nx_0)^*}{n!}=\sum_{n\ge0}\frac{1}{n!(1-nx_0)}\stackrel{{\small Treves}}{=}e\sum_{k\geq 0}B_kx_0^k
\end{eqnarray*}
where the Treves topology is just the product topology and therefore limits, for it, are computed term by term\footnote{i.e. $S_n\to S$ iff 
$$
(\forall w\in X^*)(\lim_n\scal{S_n}{w}=\scal{S}{w}).
$$
}. 

Now, it is easy to see that we have compact convergence because on $\Omega$
(or, below, $\tilde{B}$) as for all $\phi\in \calH(\Omega)$,
and $K\subset_{compact}\Omega$ (or $\tilde{B}$) one gets 
\begin{eqnarray*}
||e\sum_{k=0}^NB_k\frac{\phi^k}{k!}||_K\leq e\sum_{k=0}^NB_k\frac{||\phi||_K ^k}{k!}
\leq e\sum_{k=0}^\infty B_k\frac{||\phi||_K ^k}{k!}=ee^{e^{||\phi||_K}-1}=e^{(e^{||\phi||_K})}.
\end{eqnarray*}
Now, remarking that $\Li_{x_0^k}(z)={\log^k(z)}/{k!}$, this proves that
\begin{eqnarray*}
T\in\mathrm{Dom}(\Li_{\bullet})&\mbox{and}&\Li_T(z)=e^{(e^z)}.
\end{eqnarray*}
\end{remark}
The morphism $\Li_{\bullet}$ is no longer injective on its domain but the family $(\Li_w)_{w\in X^*}$
is still $\calC$-linearly independant \cite{orlando}.
We will use several times the following lemma which is characteristic-free.  

\subsection{Stars of the plane}

\begin{lemma}\label{diff_lemma}
Let $(\calA,d)$ be a commutative differential ring without zero divisors, and $R=\ker(d)$ be its subring of constants.
Let $z\in \calA$ such that $d(z)=1$ and $S=\{e_\alpha\}_{\alpha\in I}$ be a set of eigenfunctions of $d$,
with all different eigenvalues (for example, take $I\subset R$)
{\it i.e.}, $e_\alpha\neq0$ and $d(e_\alpha)=\alpha e_\alpha;\forall \alpha\in I$. 
Then the family $(e_\alpha)_{\alpha\in I}$ is $R[z]$-linearly free\footnote{Here $R[z]$ 
is understood as ring adjunction i.e. the smallest subring generated by $R\cup\{z\}$.}.    
\end{lemma}

\begin{proof}
If there is no non-trivial $R[z]$-linear relation, we are done. Otherwise let us consider  relations 
\begin{eqnarray}\label{lin_rel0}
\sum_{j=1}^{N}P_j(z)e_{\alpha_j}=0,
\end{eqnarray}
with\footnote{Here, $R[t]_{pol}$ means the formal univariate polynomial ring (the subscript is here to avoid confusion).}
$P_j\in R[t]_{pol},\ P_j(z)\not=0$ for all $j$ (we will, in the sequel refer to these expressions as \textit{packed linear relations}). We choose one of them minimal with respect to the triplet 
$[N,\deg(P_N),\sum\limits_{j<N}\deg(P_j)],$
lexicographically ordered from left to right\footnote{{\it i.e.} consider the ones with $N$ minimal and among these,
we choose one with $\deg(P_N)$ minimal and among these we choose one with $\sum\limits_{j<N}\deg(P_j)$ minimal.}.
Remarking that $d(P(z))=P'(z)$ (because $d(z)=1$), we apply the operator $d-\alpha_N Id$ to both sides of (\ref{lin_rel0}) and get
\begin{eqnarray}\label{lin_rel01}
\sum_{j=1}^{N}(P'_j(z)+(\alpha_j-\alpha_N)P_j(z))e_{\alpha_j}=0.
\end{eqnarray}
Minimality of (\ref{lin_rel0}) implies that (\ref{lin_rel01}) is trivial {\it i.e.}
\begin{eqnarray}\label{coeffs_rel01}
P'_N(z)=0\mbox{ and }(\forall j= 1..N-1)&&(P'_j(z)+(\alpha_j-\alpha_N)P_j(z)=0).
\end{eqnarray} 
Now (\ref{lin_rel0}) implies
\begin{eqnarray*}
\prod\limits_{j=1}^{\tcr{N}-1}(\alpha_N-\alpha_j)\sum\limits_{j=1}^N P_j(z)e_{\alpha_j}=0,
\end{eqnarray*}
which, because $\calA$ has no zero divisors, is a packed linear relation and has the same associated triplet  as (\ref{lin_rel0}).
From (\ref{coeffs_rel01}), we see that for any $k<N$, one has
\begin{eqnarray*}
\prod\limits_{j = 1}^{\tcr{N}-1}( \alpha_N-\alpha_{j})P_k(z) = \prod\limits_{j=1, j \neq k}^{\tcr{N}-1}(\alpha_N  -\alpha_j)P^{'}_k(z),
\end{eqnarray*}
so, if $N\geq 2$, we would get a relation of lower triplet. This being impossible, we get $N=1$ and (\ref{lin_rel0}) boils down to $P_N(z)e_N=0$
which, as $\calA$ has no zero divisors, implies $P_N\equiv 0$, a contradiction. Then the $(e_\alpha)_{\alpha\in I}$ are $R[z]$-linearly independent.     
\end{proof}

\begin{remark}\label{char_zero}
If $\calA$ is of characteristic zero, $d(z)=1$ implies that $z$ is transcendent
over $R$ and the two notations  $R[z]$ and $R[z]_{pol}$ coincide.
\end{remark}
First of all, let us prove 

\begin{lemma}\label{alg_ind}
Let $\mathcal{A}$ be a $\mathbb{Q}$-algebra (associative, unital, commutative) and $z$ an indeterminate, then $e^z\in \mathcal{A}[[z]]$ is transcendent over $\mathcal{A}[z]$. 
\end{lemma} 

\begin{proof}
It is a straightforward consequence of Remark \ref{char_zero}. Note that this can be rephrased as $[z,e^z]$ are algebraically independant over $\mathcal{A}$.
\end{proof}

\begin{proposition}\label{2233}
Let $Z=\{z_n\}_{n\in \N}$ be an alphabet, then $[z_0^*,z_1^*]$ is algebraically independent on $\C[Z]$ within $(\C[[Z]],\shuffle,1_{Z^*})$.   
\end{proposition}

\begin{proof}
Recall that, for $x$ a letter, one has 
\begin{eqnarray}
x^*:=\sum_{n = 0}^{+\infty}x^n=\sum_{n=0}^{+\infty}\dfrac{x^{\shuffle n}}{n!}=e^{x}_{\shuffle}.
\end{eqnarray}
By using Lemma \ref{alg_ind}, one can prove by induction that
$[e_{\shuffle}^{z_0},e_{\shuffle}^{z_1},\cdots,e_{\shuffle}^{z_k}, z_0,z_1,\cdots,z_k],$ are algebraically independent.
This implies that $Z\sqcup \{e_{\shuffle}^z\}_{z\in Z}$ is an algebraically independent set and, 
by restriction $Z\sqcup \{e_{\shuffle}^{z_0},e_{\shuffle}^{z_1}\}$ whence the proposition.  
\end{proof}

%Recall now that, for $x$ a letter, one has 
%\begin{eqnarray}
%x^*:=\sum_{n = 0}^{+\infty}x^n=\sum_{n=0}^{+\infty}\dfrac{x^{\shuffle n}}{n!}=e^{x}_{\shuffle}.
%\end{eqnarray}
Using Lemma \ref{2233}, we obtain the following proposition.
\begin{proposition}\label{*}
One has
\begin{enumerate}
\item The family $\{x_0^*,x_1^*\}$ is algebraically independent over $(\ncp{\C}{X},\allowbreak\shuffle,\allowbreak1_{X^*})$
within $(\ncs{\C}{X}^\mathrm{rat},\shuffle,1_{X^*})$.
\item The module $(\ncp{\C}{X},\shuffle,1_{X^*})[x_0^*,x_1^*,(-x_0)^*]$ is free over $\ncp{\C}{X}$ and the family 
$\{(x_0^*)^{\shuffle k}\shuffle (x_1^*)^{\shuffle l}\}_{(k,l)\in \Z\times \N}$ is a $\ncp{\C}{X}$-basis of it.
\item As a consequence, $\{w\shuffle (x_0^*)^{\shuffle k}\shuffle (x_1^*)^{\shuffle l}\}_{w\in X^*\atop (k,l)\in\Z\times \N}$ is a $\C$-basis of it.
\end{enumerate}
\end{proposition}
 
Now, we can construct the following morphism
\begin{definition}\label{nu}
The following morphism
\begin{eqnarray*}
\Li^{(1)}_{\bullet}:(\ncp{\C}{X},\shuffle,1_{X^*})[x_0^*,(-x_0)^*,x_1^*]\longrightarrow\calH(\Omega)
\end{eqnarray*}
can be defined, for any $w\in X^*$ and $\Li^{(1)}_{w}=\Li_w$, by
\begin{eqnarray*}
\Li^{(1)}_{x_0^*}=z,&\Li^{(1)}_{(-x_0)^*}=z^{-1},&\Li^{(1)}_{x_1^*}=(1-z)^{-1}.
\end{eqnarray*}
\end{definition}
%\hrule
In fact existence and uniqueness of this morphism obtained as a consequence of Proposition \ref{*}.
Moreover, its kernel and image are given by the following result\footnote{This result is already presented in \cite{ISSAC2016}, as a preprint, but never was published before.}:

\begin{theorem}\label{GH_0}
We have $\mathrm{Im}(\Li_{\bullet}^{(1)})=\calC\{\Li_w\}_{w\in X^*}$ and
$\ker(\Li_{\bullet}^{(1)})$ is the ideal generated by $x_0^*\shuffle x_1^*-x_1^*+1_{X^*}$.
\end{theorem}

\begin{proof}
As $\ncp{\C}{X}[x_0^*,x_1^*,(-x_0)^*]$ admits $\{(x_0^*)^{\shuffle k}\shuffle (x_1^*)^{\shuffle l}\}_{k\in\Z,\in \N}$
as a basis for its structure of $\ncp{\C}{X}$-module, it suffices to remark that
\begin{eqnarray*}
\Li_{(x_0^*)^{\shuffle k}\shuffle (x_1^*)^{\shuffle l}}^{(1)}(z)= \dfrac{z^k}{(1-z)^{l}},
\end{eqnarray*}
is a generating system of $\calC$. As regards the second assertion, let us prove the following lemma (in this lemma and its proof, all sums are supposed finitely supported)
\begin{lemma}\label{18012016}
Let $M_1$ and $M_2$ be $K$-modules ($K$ is a unitary commutative ring). Suppose the following map is linear
\begin{eqnarray*}
\phi: M_1\longrightarrow M_2
\end{eqnarray*}
Let $N \subset \ker(\phi)$ be a submodule. If there is a system of generators in $M_1$, namely $\{g_i\}_{i \in I}$ and $J\subset I$, such that 
\begin{enumerate}
\item For any $i\in I\setminus J$,  $g_i \equiv \sum\limits_{j \in J \subset I} c_i^j g_j \,[\mod \,N]$,  ($c_i^j \in K;  \forall j \in J$);
\item $\{\phi(g_j)\}_{j \in J}$ is $K$-free in $M_2$; 
\end{enumerate}
then $N = \ker(\phi)$.
\end{lemma}
\begin{proof}
Suppose $P \in \ker(\phi)$. Then
\begin{eqnarray*}
P \equiv\sum\limits_{j \in J}p_j g_j\,[\mod\,N]
\end{eqnarray*}
with $\{p_j\}_{j \in J}\subset K$.
Thus
\begin{eqnarray*}
0=\phi(P) =\sum\limits_{j \in J} p_j \phi( g_j ).
\end{eqnarray*}
From the fact that $\{\phi(g_j)\}_{j \in J}$ is $K-$ free in $M_2$, we obtain $p_j = 0$ for any $j\in J$.
This means that $P \in N$. Thus $\ker(\phi) \subset N$ and, finally, $N=\ker(\phi)$.
\end{proof}

\begin{figure}[htp]
\begin{center}
\definecolor{qqqqff}{rgb}{0.0,0.0,1.0}
\begin{tikzpicture}[line cap=round,line join=round,>=triangle 45,x=1.0cm,y=1.0cm]
\draw[->,color=black] (-3.0,0.0) -- (3.0,0.0);
\foreach \x in {-2.0,-1.0,1.0,2.0}
\draw[shift={(\x,0)},color=blue] (0pt,2pt) -- (0pt,-2pt);
\draw[->,color=black] (0.0,0.0) -- (0.0,3.0);
\foreach \y in {0.0,1.0,2.0}
\draw[shift={(0,\y)},color=red] (2pt,0pt) -- (-2pt,0pt);
\clip(-3.0,0.0) rectangle (3.0,3.0);
%\draw[color=qqqqff,smooth,samples=100,domain=-4.0:0.9] plot(\x,{(2.0*(\x)+1.0)/((\x)-1.0)});
%\draw[color=qqqqff,smooth,samples=100,domain=1.1:6.0] plot(\x,{(2.0*(\x)+1.0)/((\x)-1.0)});
%\draw[color=qqqqff,domain=0.0:1.0] plot(\x,{(--2.0-0.0*\x)/1.0});
%\draw[color= blue,domain=-3.0:3.0] plot(\x,{(--0.02-0.0*\x)/1.5});
%\draw [domain=-4.0:6.0] plot(\x,{(--2.0-0.0*\2)/1.0});
\begin{scriptsize}
%\draw[color=qqqqff] (-2.46,0.62) node {$f(x) = (2x + 1) / (x - 1)$};
%\draw [fill=black] (1.0,2.0) circle (1.0pt);
%\draw[color=black] (1.0,2.0) node {$\cdot$};
\draw[color=black] (2.2,2.2) node {$(w, \textcolor[rgb]{0.0,0.0,1.0}{l},\textcolor[rgb]{1.0,0.0,0.0}{k}) $};
%\draw[color=black] (1.0,2.0) node {$\circlearrowright$};
\draw[color=black] (0.0,0.0) node {$\bigcirc$};
\draw[color=black] (-0.0,0.2) node {$(1_{X^*}, \times, \times )$};
\draw[color=red] (-0.2,1.8) node {$k$};
\draw[color=black] (0.0,2.0) node {$\cdot$};
%\draw[color=blue] (1.1,0.2) node {$l$};
\draw[color=black] (1.0,0.0) node {$\cdot$};
%\draw[color=black] (-1.0,1.0) node {$\cdot$};
\draw[color=black] (-2.1,2.2) node {$(w, \textcolor[rgb]{0.0,0.0,1.0}{-l},\textcolor[rgb]{1.0,0.0,0.0}{k})$};
%\draw[color=black] (-1.0,1.0) node {$\circlearrowleft$};
%\draw[color=red] (-0.3,2.5) node {$k$};
%\draw[color=red] (-0.15,2.7) node {$$};
\draw[color=blue] (2.0, 0.5) node {$l$};
%\draw[color=blue] (2.7,0.2) node {$l$};
\draw[color=blue] (-2.0,0.2) node {$-l$};
\draw[red] (1,2)--(2,2);
\draw[red] (1,1)--(2,2);
\draw[color=black] (1,2) node {$\triangleleft$};
\draw[color=black] (1,1) node {$\triangleright$};
\draw[color=black] (0.7,2.2) node {$(w, \textcolor[rgb]{0.0,0.0,1.0}{l-1},\textcolor[rgb]{1.0,0.0,0.0}{k})$};
\draw[color=black] (0.7,0.8) node {$(w, \textcolor[rgb]{0.0,0.0,1.0}{l-1},\textcolor[rgb]{1.0,0.0,0.0}{k - 1})$};
\draw[blue] (-1,2)--(-2,2);
\draw[blue] (-2,1)--(-2,2);
\draw[color=black] (-1,2) node {$\triangleright$};
\draw[color=black] (-2,1) node {$\triangledown$};
\draw[color=black] (-0.8,2.2) node {$(w, \textcolor[rgb]{0.0,0.0,1.0}{-l+1},\textcolor[rgb]{1.0,0.0,0.0}{k})$};
\draw[color=black] (-2.2,1.2) node {$(w, \textcolor[rgb]{0.0,0.0,1.0}{-l},\textcolor[rgb]{1.0,0.0,0.0}{k-1})$};
\end{scriptsize}
\end{tikzpicture}
\end{center}
\caption{Rewriting $\mod {\cal J}$ of  $\{w\shuffle(x_0^*)^{\shuffle l}\shuffle(x_1^*)^{\shuffle k}\}_{k\in\N,l \in\Z,w\in X^*}$.}
\label{Picture of basis}
\end{figure}
Let now ${\cal J}$ be the ideal generated by $x_0^*\shuffle x_1^*-x_1^*+1_{X^*}$.
It is easily checked, from the following formulas\footnote{In Figure \ref{Picture of basis},
$(w, l, k)$ codes the element $w\shuffle(x_0^*)^{\shuffle l} \shuffle (x_1^*)^{\shuffle k}$.},  for $l\geq 1$,
\begin{eqnarray*}
w \shuffle x_0^* \shuffle (x_1^*)^{\shuffle l}&\equiv& w \shuffle
(x_1^*)^{\shuffle l}-w \shuffle(x_1^*)^{\shuffle l-1}\ [{\cal J}],\cr
w \shuffle (-x_0)^*\shuffle (x_1^*)^{\shuffle l}&\equiv& w \shuffle
(-x_0)^*\shuffle (x_1^*)^{\shuffle l-1}+w \shuffle(x_1^*)^{\shuffle l}[{\cal J}],    
\end{eqnarray*} 
that one can rewrite $[\mod\, {\cal J}]$ any monomial $w \shuffle(x_0^*)^{\shuffle k}\shuffle(x_1^*)^{\shuffle l}$ as
linar combination of such monomials with $kl=0$. Applying Lemma \ref{18012016} with the following data:
\begin{itemize}
\item the morphism  $\phi=\Li^{(1)}_{\bullet}$, 
\item the modules $M_1=\ncp{\C}{X}[x_0^*,x_1^*,(-x_0)^*], M_2=\calH(\Omega),N={\cal J}$,
\item the families $\{g_i\}=\{w\shuffle(x_1^*)^{\shuffle n}\shuffle(x_0^*)^{\shuffle m}\}_{(w, n,m) \in I},$ 
\item and the indices $I=X^*\times\N\times\Z$,  $J=(X^*\times\N\times\{0\})\sqcup(X^*\times \{0\} \times \Z)$,
\end{itemize}

 we have the second point of Theorem \ref{GH_0}.
\end{proof}

\subsection{Examples of polylogarithms indexed by rational series}

\begin{proposition}[\cite{TCS,IMACS}]\label{exemples}
One has
\begin{enumerate}
\item For $x\in X,i\in\N_+,a\in{\mathbb C},\abs{a}<1$,
\begin{eqnarray*}
\Li^{(1)}_{(ax_0)^{*i}}(z)&=&z^a\Sum_{k=0}^{i-1}{i-1\choose k}\Frac{(a\log(z))^k}{k!},\\
\Li^{(1)}_{(ax_1)^{*i}}(z)&=&\Frac1{(1-z)^a}\Sum_{k=0}^{i-1}{i-1\choose k}\Frac{(a\log((1-z)^{-1})^k}{k!}.
\end{eqnarray*}

\item For any $(s_1,\ldots,s_r)\in\N _+^r$ and $|t_1|<1,\ldots,|t_r|<1$,
\begin{eqnarray*}
\Li^{(1)}_{(t_1x_0)^{*s_1}x_0^{s_1-1}x_1\ldots(t_rx_0)^{*s_r}x_0^{s_r-1}x_1}(z)
=\sum_{n_1>\ldots>n_r>0}\frac{z^{n_1}}{(n_1-t_1)^{s_1}\ldots(n_r-t_r)^{s_r}}.
\end{eqnarray*}
In particular,
\begin{eqnarray*}
\Li^{(1)}_{(t_1x_0)^*x_1\ldots(t_rx_0)^*x_1}(z)=\sum_{n_1,\ldots,n_r>0}
\Li_{x_0^{n_1-1}x_1\ldots x_0^{n_r-1}x_1}(z)\;t_0^{n_1-1}\ldots t_r^{n_r-1}.
\end{eqnarray*}
\end{enumerate}
\end{proposition}

To prove this proposition, we use the following easy lemma:

\begin{lemma}\label{lemmaHoan}
For any $i, n \in \mathbb{N}^* $, we have
\begin{eqnarray*}
\binom{n+i-1}{n} = \sum\limits_{k=0}^{n}\binom{i-1}{k} \binom{n}{n-k}.
\end{eqnarray*}
\end{lemma}
%\begin{proof}
%The proof is left to the reader. 
%\end{proof}

Now, we give the proof of Proposition \ref{exemples}.
\begin{proof}
\begin{enumerate}
\item
Let's choose $i\in\N^*,a\in\C$ and $x \in X$. Note that
\begin{eqnarray*}
(ax)^* &=& \sum\limits_{n =0}^{\infty} (ax)^n = \dfrac{1}{1-ax},\\
((ax)^*)^i &=& \sum\limits_{n=0}^{\infty} \binom{n+i-1}{n} (ax)^n\ ,\\
(ax)^* \shuffle (1+ax)^{i-1} &=& \sum\limits_{n = 0}^{\infty} (\sum\limits_{k=0}^n \binom{i-1}{k}\binom{n}{k})(ax)^n, 
\end{eqnarray*}
and use Lemma \ref{lemmaHoan}, we obtain, for $i\in\N^*,a\in\mathbb{C},x\in X$,
\begin{eqnarray*}
((ax)^*)^i=(ax)^*\shuffle(1+ax)^{i-1}.
\end{eqnarray*}
Thus, for $i\in\N^*,|a|<1,x\in X$, 
\begin{eqnarray*}
\Li^{(1)}_{((ax)^*)^i}
=\Li^{(1)}_{(ax)^*}\Li^{(1)}_{(1+ax)^{i-1}}
=\Li^{(1)}_{(ax)^*}\sum\limits_{k=0}^{i-1}\binom{i-1}{k}a^k\Li^{(1)}_{x^k}.
\end{eqnarray*}
It follows then the expected results.
\item Using Lemma \ref{lemmaHoan}, we obtain this statement by direct calculation.
\end{enumerate}
\end{proof}

\begin{corollary}\label{cor1}
One has
\begin{eqnarray*} 
\{\Li_S\}_{S\in\CX\shuffle\C[x_0^*]\shuffle\C[(-x_0^*)]\shuffle\C[x_1^*]}
&=&\mathrm{span}_\C\biggl\{\Frac{z^a}{(1-z)^b}\Li_w(z)\biggr\}_{w\in X^*}^{a\in\Z,b\in\N}\cr
&\subset&\mathrm{span}_{\C}\{\Li_{s_1,\ldots,s_r}\}_{s_1,\ldots,s_r\in\Z^r}\\
&&\oplus\mathrm{span}_{\C}\{z^a|a\in\Z\},\\
\{\Li_S\}_{S\in\CX\shuffle\serie{\C^{\mathrm{rat}}}{x_0}\shuffle\serie{\C^{\mathrm{rat}}}{x_1}}
&=&\mathrm{span}_\C\biggl\{\Frac{z^a}{(1-z)^b}\Li_w(z)\biggr\}_{w\in X^*}^{a,b\in\C}\cr
&\subset&\mathrm{span}_{\C}\{\Li_{s_1,\ldots,s_r}\}_{s_1,\ldots,s_r\in\C^r}\cr
&&\oplus\mathrm{span}_{\C}\{z^a|a\in\C\}.
\end{eqnarray*}
\end{corollary}

\section{Symmetries and transition elements}\label{trans}
\subsection{Framework}\label{Symm}
Up to now (and historically \cite{LD}), the polylogarithms are computed in $\Omega=\C\setminus(]-\infty,0]\cup[1,+\infty[)$, cleft in order to cope with the two singularities $\{0,1\}$. But 
$B=\C\setminus \{0,1\}$ is acted on by the following group of symmetries (which permutes, in fact, 
$\{0,1,+\infty\}$). 
\begin{eqnarray*}
{\cal G}:=\{z\mapsto z,z\mapsto 1-z,z\mapsto z^{-1},z\mapsto(1-z)^{-1},z\mapsto 1-z^{-1},z\mapsto z(z-1)^{-1}\}
\end{eqnarray*}
To this end and because $\C\setminus(]-\infty,0]\cup[1,+\infty[)$ is not stable by this group we have now to work on $\Omega=\tilde{B}$.
\subsection{Monodromy Principle}
For convenience, we consider the following situation 
$$\begin{tikzcd}
&Y\ar{d}{p}\\
X'\ar{r}{f}\ar[dashed]{ur}{g}& X
\end{tikzcd}$$
and recall the \textit{monodromy principle} (see \cite{D3} 16.28.8)
\begin{theorem}[Monodromy Principle]\label{MP}
Let $(Y,X,p)$ be a covering of a differential manifold $X$
and let $f : X'\to X$ be a $C^{\,\infty}$-mapping of a simply-connected\footnote{
Nowadays simply-connected implies path-connected.}, differential manifold.
Let $a'\in X'$ and $b\in p^{-1}(f(a'))$. Then there exists a unique
$C^{\infty}$-mapping $g:X'\to Y$ such that $g(a')=b$ and
$p\circ g =f$ (the mapping $g$ is said to be a lifting of $f$). 
\end{theorem}
Here we specialize this to $B=\C\setminus\{0,1\}$ choose a universal covering $(B,\widetilde{B},p)$
and a section $s:\Omega\to\widetilde{B}$ of $p$, lifted from the canonical embedding $j:\Omega\hookrightarrow B$
$$\begin{tikzcd}
&\widetilde{B}\ar{d}{p}\\
\Omega\ar[hook]{r}{j}\ar[hook,dashed]{ur}{s}&B
\end{tikzcd}$$
We first remark that any $g\in {\cal G}$ maps in fact $B=\C\setminus\{0,1\}$
to itself and apply the Monodromy Principle to the following situation
$$\begin{tikzcd}
\widetilde{B}\ar{d}{p}\ar[dashed]{r}{\tilde{g}}&\widetilde{B}\ar{d}{p}\\
B\ar{r}{g}&B
\end{tikzcd}$$
where $(\widetilde{B},B,p)$ is any universal covering of $B$.

First remark that ${\cal G}$ is a copy of  $\mathfrak{S}_3$ as it permutes the three singularities
\begin{figure}[h!]%\label{diag3}
$$\begin{tikzcd}
&10\,\infty\ar{dl}{1-z}\ar{dr}{1/z}&\\
01\,\infty \ar{d}{1/z}&&1\,\infty\,0\ar{d}{1-z}\\
\infty\,10\ar{dr}{1-z}&&0\,\infty\,1\ar{dl}{1/z}\\
&\infty\,01&
\end{tikzcd}$$
\caption{Orbit of the singularities}
\end{figure}
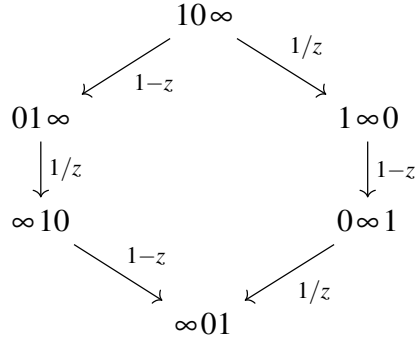
and choose an orbit in $\Omega$ (as the orbit of $i$), for instance
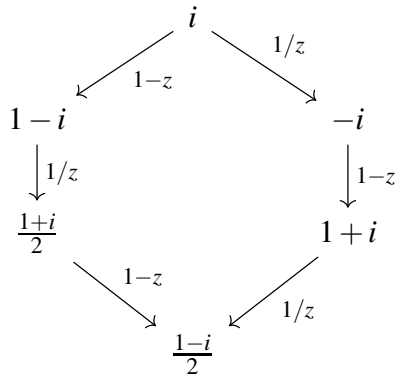
\begin{figure}[h!]
$$\begin{tikzcd}
\ & i\ar{dl}{1-z}\ar{dr}{1/z} &\ \\
1-i\ar{d}{1/z}& & -i\ar{d}{1-z}\\
\frac{1+i}{2}\ar{dr}{1-z} & & 1+i\ar{dl}{1/z}\\
\ & \frac{1-i}{2} &\
\end{tikzcd}$$
\caption{Orbit of $i$}
\end{figure}
now, we can pinpoint the lifting (i.e. find points like $a',b$ in Theorem \ref{MP}).

Let us note $g_{1-z}(z)=1-z,g_{1/z}(z)=z^{-1}\in\mathfrak{S}_B$ and set 
$\widetilde{g}_\phi\in \mathfrak{S}_{\widetilde{B}}$ such that 
\begin{equation*}
p\circ\widetilde{g}_\phi=g_\phi\circ p\mbox{ with }p\circ\widetilde{g}_\phi(s(z_0))=g_\phi(z_0).
\end{equation*}
where $z_0=i$.
\begin{remark} An involutive bi-homomorphism cannot always be lifted as a permutation of finite order
as shows the example of $\Z/2\Z$ acting on $\C^*=\C\setminus \{0\}$ by $a(z)=-z$.
It can be shown (exercise left to the reader) that every lifting $\tilde{a}$ of $a$  i.e. 
$$\begin{tikzcd}
\widetilde{\C^*} \ar{d}{p}\ar[dashed]{r}{\tilde{a}} & \widetilde{\C^*}\ar{d}{p}\\
\C^*\ar{r}{a} & \C^*
\end{tikzcd}$$
is of infinite order.

%In the same vein, using $\C-[0,1]$, instead of $\C^*$, 
%we can see that $\widetilde{g}_{1-z}$ is of infinite order.  
\end{remark}
Before computing the transition maps for $\widetilde{\L}$ (extension of $\L$ to $\widetilde{B}$) under 
 $\widetilde{g}_{1-z}$ and $\widetilde{g}_{1/z}$, we must take an excursion to noncommutative differential equations. 

\subsection{Noncommutative differential equations}
Let $(V,d)$ be a one-dimensional $C^\omega$ connected complex manifold\footnote{A Riemann surface in short.}
with its derivation $d={d}/{dz}$. We endow $\ncs{\calH(V)}{X}$ with $\dd$
\begin{eqnarray}
\dd(S)=\sum_{w\in X^*}d(\scal{S}{w})w\ .
\end{eqnarray}
It can be easily checked that $\dd$ is a derivation of the $\C$-algebra $\ncs{\calH(V)}{X}$. 

We are now able to define noncommutative differential equations (left multiplier case). 
\begin{definition} A noncommutative differential equation, on $V$, with left multiplier $M\in \ncs{\calH(V)_+}{X}$ is an equality
\begin{eqnarray}\label{NCdiff}
\dd(S)=MS.
\end{eqnarray}
An initial condition can be pinpointed (with $z_0\in V$ and $S_0\in \ncs{\C}{X}$)
\begin{eqnarray}\label{EDL_initial}
\left\{\begin{array}{c}
\dd(S)=MS,\\
S(z_0)=S_0.
\end{array}\right.
\end{eqnarray}
It can be asymptotic\footnote{Below and in general $\mathfrak{F}$ is a filter, the reader who is not familiar with these objects can replace $\mathfrak{F}$ by any \textit{mode of convergence with respect to a subset (e.g. a cone, cluster point, to infinity - full or with restrictions - etc.).}.}
\begin{eqnarray}\label{EDL_initial_asymptotic}
\left\{\begin{array}{c}
\dd(S)=MS,\\
\lim_{\mathfrak{F}}S(z)=S_0.
\end{array}\right.
\end{eqnarray}
\end{definition}
We gather here the needed results 
\begin{proposition}[\cite{Linz,CASC2018}]\label{ED_prop}
We have the following properties:
\begin{enumerate}
\item If $V$ is simply connected, equation (\ref{EDL_initial}) has a unique solution:
\begin{itemize}
\item with $S_0=1_{X^*}$, it can be computed, through Picard's process, by iterated integrals with lower bond $z_0$, this solution will be noted  $S_{Pic}^{z_0}$
\item in the general case (initial condition $S_0$) the solution is $S_{Pic}^{z_0}S_0$.
\end{itemize} 

\item If $V$ is connected, solutions to equation (\ref{EDL_initial_asymptotic}) may not exist, but if it does, the solution is unique.
\item ($V$ is connected) The set of solutions of (\ref{NCdiff}) is a vector space. Two solutions which coincide at a point - actual (\ref{EDL_initial}) or asymptotic (\ref{EDL_initial_asymptotic}) - coincide everywhere.
\item ($V$ is simply connected) The set of invertible solutions of (\ref{EDL_initial}) is the following orbit on the right
\begin{eqnarray*}
\calS=S_{Pic}^{z_0}(\ncs{\C}{X})^\times.
\end{eqnarray*} 
\end{enumerate}
\end{proposition}
\subsection{Equivariance of polylogs on $\widetilde{B}$.} 

Now, we will explain the property of polylogs expressed by formulas of the type
\begin{eqnarray*}
\widetilde{\L}(g.z)=\mu_g(\L(z))Z(g),
\end{eqnarray*}
where $\mu_g$ is a morphism of alphabets of the type 
\begin{eqnarray}
\begin{pmatrix}
\mu_g(x_0)\\
\mu_g(x_1)
\end{pmatrix}
=
\begin{pmatrix}
a_{11} & a_{12}\\
a_{21} & a_{22}
\end{pmatrix}
\begin{pmatrix}
x_0\\
x_1
\end{pmatrix}
\end{eqnarray}
and $Z(g)\in \ncs{\C}{X}$.

We detail here the computation for $\widetilde{g}_{1-z}$. 
\begin{enumerate}
\item $\widetilde{g}_{1-z}$ has been choosen such that 
\begin{eqnarray}
p\circ \widetilde{g}_{1-z}=g_{1-z}\circ p&\mbox{and}&p\circ \widetilde{g}_{1-z}(s(i))=1-i.
\end{eqnarray}
\item From $p:\widetilde{B}\to B$, one has $p^*:\calH(B)\to\calH(\widetilde{B})$ and define the functions 
$\widetilde{\phi}_z$ (resp. $\widetilde{\phi}_{1/z}$) in $\calH(\widetilde{B})$ by  
\begin{eqnarray}
\widetilde{\phi}_z(u)=p(u)\in\C&\mbox{resp.}&\widetilde{\phi}_{1/z}(u)=\frac{1}{p(u)}\in\C.
\end{eqnarray}
\item Solve, in $\ncs{\calH(\widetilde{B})}{X}$,
\begin{eqnarray}\label{EDL_uc_asymptotic}
\left\{\begin{array}{c}
\dd(\widetilde{S})(u)=\Big(\Frac{x_0}{\widetilde{\phi}_z(u)}+\Frac{x_1}{\widetilde{\phi}_{1-z}(u)}\big)\widetilde{S}(u),\\
\lim_{z\to 0\atop z\in \Omega}\widetilde{S}(s(z))e^{-x_0\,\log(z)}=1_{X^*}.
\end{array}\right.
\end{eqnarray}
The uniqueness of such a solution is a result of Proposition \ref{ED_prop} (iii).
The existence can be obtained by lifting or remarking that the analog of \eqref{EDL_uc_asymptotic}
(first row) can be solved at the level of $\Omega$ and, for any choice of $z_0\in \Omega$, one has 
\begin{eqnarray}
\widetilde{S}_{Pic}^{s(z_0)}\circ s=S_{Pic}^{z_0}
\end{eqnarray}
then the solution of (\ref{EDL_uc_asymptotic}) is given by
\begin{eqnarray*}
\widetilde{\L}=\widetilde{S}_{Pic}^{s(z_0)}\L(z_0)
\end{eqnarray*}
and satisfies
\begin{eqnarray*}
\widetilde{\L}\circ s=\L.
\end{eqnarray*}
\item Now
\begin{eqnarray*}
\dd(\widetilde{\L}\circ \widetilde{g}_{1-z})\circ s=\frac{d}{dz}(\widetilde{\L}\circ \widetilde{g}_{1-z})\circ s
\end{eqnarray*}
and, in a neighbourhood of $i$, one has 
\begin{eqnarray*}
\widetilde{\L}\circ \widetilde{g}_{1-z}\equiv\L(1-z).
\end{eqnarray*}
\item Setting $\L_1(z)=\L(1-z)$, one gets 
\begin{eqnarray*}
\frac{d}{dz}\L_1=\Big(-\frac{x_0}{1-z}-\frac{x_1}{z}\Big)\L_1=\mu\Big(\frac{x_0}{z}+\frac{x_1}{1-z}\Big)\L_1,
\end{eqnarray*}
where $\mu$ is the morphism of alphabets 
\begin{eqnarray*}
\mu\begin{pmatrix}
x_0\\
x_1
\end{pmatrix}
=
\begin{pmatrix}
0 & -1\\
-1 & 0
\end{pmatrix}
\begin{pmatrix}
x_0\\
x_1
\end{pmatrix}
\end{eqnarray*}
\item as $\mu$  permutes with the derivation and is invertible, one has
\begin{eqnarray*}
\frac{d}{dz}\mu^{-1}(\L_1)=\Big(\frac{x_0}{z}+\frac{x_1}{1-z}\Big)\mu^{-1}(\L_1).
\end{eqnarray*}
Hence,
\begin{eqnarray*}
\mu^{-1}(\L_1)=\L Z_1&\mbox{and}&\L_1=\mu(\L)\mu(Z_1).
\end{eqnarray*}
\item Finally, as they coincide on an open subset and are analytic, one gets
\begin{eqnarray*}
\widetilde{\L}_1=\mu(\widetilde{\L})\mu(Z_1).
\end{eqnarray*}
\end{enumerate}

\section{Applications}\label{applications}
In this section, we will give some applications of this new presentation of polylogarithms at negative integer multi-indices.
\subsection{A new presentation of harmonic sums with non-positive multi-indices}
From Corollary \ref{cor1}, for any $w \in Y^*_0$, the corresponding polylogarithm
$\Li^-_{w}$ is an element of the algebra $\mathbb{Q}[(1-z)^{-1}]$.
Thus, for $w\in Y_0^*$, we suppose that $\Li_{w}^-$ can be expanded as follows
\begin{eqnarray}\label{gor1}
\Li_w^-(z)=\Sum_{k=0}^{(w)+|w|}\Frac{a^w_k}{(1-z)^k}&\mbox{with}&a^w_i\in\Q.
\end{eqnarray}
We note that $a_{0}^{1_{Y^*_0}}=1$ and, for any $n > 0$,
$a_{n}^{1_{Y^*_0}}=0$.
Then, using Proposition \ref{propthu}, the sequences $\{a_{i}^w\}_{w\in Y^+_0, n \in \N}$ are computed as follows:

\begin{lemma}[Algorithm to compute $\Li^-_{\bullet}$]\label{presentation}
\item \quad Let $w=y_ku \in Y_0^+=Y_0^*\setminus \{1_{Y_0^*}\}$\footnote{Such a (non-empty) word $w$ can always be written uniquely $w=y_ku$ where $y_k\in Y_0$ and $u\in Y_0^*$.}. We have
\begin{enumerate}
\item If $k= 0$ then
\begin{eqnarray*}
a^{y_0u}_i := \begin{cases} a_{i-1}^u & \mbox{ for } i = (u)+ |u|+1,\\
a_{i-1}^u - a_{i}^u & \mbox{ for } 1 \leq i \leq (u) +|u|,\\
- a_{i}^u & \mbox{ for } i = 0,\\
0& \mbox{otherwise}.\\
 \end{cases} 
\end{eqnarray*}
\item If $k > 0$ then 
\begin{eqnarray*}
a^{y_ku}_i := \begin{cases} (i-1)a_{i-1}^{y_{k-1}u} & \mbox{ for } i = (u)+ |u|+k+1,\\
(i-1)a_{i-1}^{y_{k-1}u} - i a_{i}^{y_{k-1}u} & \mbox{ for } 2 \leq i \leq (u) +|u|+k,\\
- a_{i}^{y_{k-1}u} & \mbox{ for } i = 1,\\
0& \mbox{for otherwises}.\\
\end{cases} 
\end{eqnarray*}
\end{enumerate}
\end{lemma}

\begin{proof}
\begin{enumerate}
\item If $k=0$, we have
$\Li^-_{y_0u} = \Li^-_{y_0} \Li_u^-.$
Hence, we obtain this condition.
\item This condition is a direct corollary of the identity $\Li^-_{y_ku}=\theta_0\Li^-_{y_{k-1}u}.$
\end{enumerate}
\end{proof}

\begin{proposition}\label{corol5}
For any $w\in Y^*_0$, we have $\Li_w^-=\Li_{P_w}$,
where
\begin{eqnarray*}
P_w:=\Sum_{i=0}^{(w)+|w|}a_i^w (x_1^*)^{\shuffle i}\in(\C[x_1^*],\shuffle,1_{X^*})
\end{eqnarray*}
and the coefficients $\{a_i^w\}_{i \in \N}$ are defined as in Lemma \ref{presentation}.
\end{proposition}

\begin{table}[!ht]
\centering
\renewcommand{\arraystretch}{1.25}
\small{\begin{tabular}{|c|c|c|c|c|c|c|c|c|c|c|}
\hline
$k \setminus i$ & $0$ &$1$ &$\colorbox[rgb]{0.50,1.0,0.50}{2}$ &$\colorbox[rgb]{0.50,1.0,0.50}{3}$ &$\colorbox[rgb]{1.0,0.0,1.0}{4}$ &$5$ &$6$ &$7$& $\ldots$\\ 
\hline
0&1 &\colorbox[rgb]{1.0,1.0,0.0}{-2}&1 &0&0 & 0& 0&0& $\ldots$\\ 
1&0 & \colorbox[rgb]{1.0,1.0,0.0}{2}& \colorbox[rgb]{0.50,1.0,0.50}{-4} & \colorbox[rgb]{0.50,1.0,0.50}{2}&0 & 0& 0&0& $\ldots$\\ 
2&0 & -2&10 &\colorbox[rgb]{0.50,1.0,0.50}{-14}&\colorbox[rgb]{1.0,0.0,1.0}{6}& 0& 0&0& $\ldots$\\ 
3&0 & 2& -22 &62&-66& \colorbox[rgb]{1.0,0.0,1.0}{24}& 0&0& $\ldots$\\ 
4&0 & -2& 46 & -230&450& -384& 120&0& $\ldots$\\ 
$\ldots$&$\ldots$ & $\ldots$& $\ldots$ & $\ldots$&$\ldots$& $\ldots$& $\ldots$&$\ldots$& $\ldots$\\
 \hline
\end{tabular}}
\caption{The value of $a^{y_ky_0}_i$}\label{tab:*}
\end{table}

Now, we note that for any $k \in \N$, 
\begin{eqnarray*}
\Frac{1}{(1-z)^{k+1}}=\Frac{1}{k!}\Sum_{n\geq k}n\ldots(n-k+1)z^{n-k}=\Sum_{N\ge0}\binom{N+k}{k}z^N.
\end{eqnarray*}
On the other hand, for any $w \in Y^*_0$, denoting $p=(w)+|w|$, we have
\begin{eqnarray*}
\Sum_{N\ge0}\H^-_w(N) z^N=\Frac{\Li^-_w(z)}{1-z}=\Frac{a^w_{p}}{(1-z)^{p+1}}+\ldots+\Frac{a^w_1}{(1-z)^2}+\Frac{a^w_0}{1-z}. 
\end{eqnarray*}
Hence,
\begin{eqnarray*}
\Sum_{N \geq 0}\H^-_w(N) z^N=a^w_{p}\Sum_{N \geq 0} \binom{N+p}{p} z^N+\ldots+a^w_{0}\Sum_{N \geq 0}\binom{N}{0} z^N. 
\end{eqnarray*}

This means that, 
\begin{proposition}\label{important}
For any $w \in Y^*_0$,
\begin{eqnarray*}
\H^-_w(N)=\Sum_{k=0}^{(w)+|w|}a^w_{k}\binom{N+k}{k}.
\end{eqnarray*}
\end{proposition}

In fact, since the definition of the sequence $\{a^w_{p}\}_{w\in Y_0^*, p \in \N}$, we can use the formula
\begin{eqnarray*}
\dbinom{n+r}{r}-\dbinom{n+r-1}{r-1}=\dbinom{n+r-1}{r}
\end{eqnarray*}
and Proposition \ref{important} to obtain an extension of Faulhaber's result \cite{Faul},
{\it i.e.,} the harmonic sum $\H^-_w(N), w\in Y^*_0$ can be written like
a finite linear combination of the elements in the family $\{\dbinom{N+n}{m}\}_{n,m\in\N}$,
where the number of terms is at most $\lfloor\dfrac{(w)+|w|}{2}\rfloor +1$. 
\begin{example}

\begin{eqnarray*}
\H^-_{y_0}(N) &=& \binom{N+1}{1} - \binom{N}{0}\cr &=& \binom{N}{1},\cr
\H^-_{y_1}(N) &=& \binom{N+2}{2} - \binom{N+1}{1}\cr &=& \binom{N+1}{2},\cr
\H^-_{y_2}(N) &=& 2\binom{N+3}{3} - 3\binom{N+2}{2} + \binom{N+1}{1}\cr &=& 2 \binom{N+2}{3} - \binom{N+1}{1},\cr
\H^-_{y_3}(N) &=& 6\binom{N+4}{4} -12 \binom{N+3}{3} + 7\binom{N+2}{2} - \binom{N+1}{1}\cr &=& 6\binom{N+2}{4} + \binom{N+1}{2},\cr
\H^-_{y_0^2}(N) &=& \binom{N+2}{2} - 2\binom{N+1}{1} + \binom{N}{0}\cr &=& \binom{N}{2},\cr
\H^-_{y_1y_0}(N)&=& 2 \binom{N+3}{3} - 4 \binom{N+2}{2} +2\binom{N+1}{1}\cr &=& 2 \binom{N+1}{3},\cr
\H^-_{y_2y_0}(N) &=& 6\binom{N+4}{4} - 14\binom{N+3}{3} +10\binom{N+2}{2} -2 \binom{N+1}{1}\cr &=& 6\binom{N+2}{4} + 2 \binom{N+1}{2}.
\end{eqnarray*}
\end{example}

\subsection{Regularization of polyzetas at negative multi-indices} 
Let $\{t_i\}_{i\in \N_+}$ be a family of variables.
The symmetric functions $\{\eta_k\}_{k \in \N_+}$ and the power sums
$\{\psi_k\}_{k\in \N_+}$ are defined \cite{GKLLRT} respectively by 
\begin{eqnarray*}
\eta_k(\underline{t}):=\Sum_{n_1>\ldots n_k>0} t_{n_1}\ldots t_{n_k}&\mbox{and}&\psi_k(\underline{t}):=\Sum_{n>0}t_n^k.
\end{eqnarray*}   
Then the generating series of the family\footnote{We set $\eta_0(\underline{t}) = 1$.} $\{\eta_k\}_{k\in \N}$ is defined by
\begin{eqnarray*}
1+\Sum_{k \geq 1}\eta_k(\underline{t})z^k=\prod_{i>0}(1+ t_i z)=:\eta(\underline{t}\mid z).  
\end{eqnarray*}
In the same way, we also define
\begin{eqnarray*}
\Sum_{k \geq 1}\psi_k(\underline{t})z^k=\Sum_{i\ge1}\frac{t_iz}{1-t_iz}=:\psi(\underline{t}\mid z).  
\end{eqnarray*}
Note that the functions $\eta (\underline{t}\mid z)$ and $\psi (\underline{t}\mid z)$ satisfy Euler's identity
\begin{eqnarray}\label{equs}
z\frac{d}{dz}\log(\eta(\underline{t}\mid z))=\psi(\underline{t}\mid -z).
\end{eqnarray}

For any $w=y_{s_1} \ldots y_{s_r}\in Y^*$, the quasi-symmetric function of depth $|w|=r$
and of degree (or weight) $(w):=s_1+\ldots+s_r$,namely $\F_w$, is defined by
\begin{eqnarray*}
\F_w(\underline{t}):=\Sum_{n_1>\ldots>n_r>0}t_{n_1}^{s_1}\ldots t_{n_r}^{s_r}.
\end{eqnarray*}
Note that $\F_{y_1^k}=\eta_{k}$ and $\F_{y_k}=\psi_k$, and from equation \eqref{equs}, we obtain 
\begin{eqnarray}\label{6.1}
\Sum_{k\ge0}\F_{y_1^k}z^k=\exp\Big(-\Sum_{k\ge1}\F_{y_k} \Frac{(-z)^k}{k}\Big).
\end{eqnarray}
Remark that the harmonic sum $\H_w(N)$ can be obtained by specializing,
in the quasi-symmetric function $\F_w$,
the variables $\{t_i\}_{i\ge1}$ as follows \cite{words03}:
\begin{eqnarray*}
t_i=\frac{1}{i}\quad\mbox{if}\quad0<i\le N\quad\mbox{else}\quad t_i=0.
\end{eqnarray*}

On the other hand, we recall the morphisms of regularization \cite{JSC}
\begin{eqnarray*}
\gamma_{\bullet}:(\ncp{\C}{Y},\stuffle,1_{Y^*})&\longrightarrow&(\C,\cdot, 1),\\
\zeta_{\shuffle}:(\ncp{\C}{X},\shuffle,1_{X^*})&\longrightarrow&(\C,\cdot, 1),
\end{eqnarray*}
which are defined by
\begin{eqnarray*}
\gamma_{y_1}=\gamma,&&\zeta_{\shuffle}(x_0)=\zeta_{\shuffle}(x_1)=0
\end{eqnarray*}
and
\begin{eqnarray*}
\forall l\in\Lyn(X)-X,&&\gamma_{\pi_{Y}(l)}=\zeta(l).
\end{eqnarray*}
Then
\begin{eqnarray*}
\forall t\in\C,&&\zeta_{\shuffle}((tx_1)^*)=1.
\end{eqnarray*}
Note that, by a Newton-Girard like identity (see \eqref{6.1}), one has \cite{JSC}
\begin{eqnarray*}
\Sum_{k \geq 0}\H_{y_1^k}z^k=\exp\Big(-\Sum_{k\geq 1}\H_{y_k}\Frac{(-z)^k}{k}\Big),
\end{eqnarray*}
and by the properties of Gamma function, we obtain, for $|t|<1$, 
\begin{eqnarray}
\gamma_{\pi_{Y}((tx_1)^*)}=\exp\Big(\gamma t-\Sum_{n\geq 2}\zeta(n)\frac{(-t)^n}{n}\Big)=\frac{1}{\Gamma(1+t)}.
\end{eqnarray}
For $t\in\N$, we put 
\begin{eqnarray*}
\gamma_{\pi_{Y}((tx_1)^*)}=\frac{1}{t!}. 
\end{eqnarray*}

\begin{proposition}[\cite{inpreparation}]\label{singuttt}
Let $w = y_{s_1} \ldots y_{ s_r}\in Y_0^*$. Then there exists a \underline{unique}
polynomial $P_w\in(\C[x_1^*],\shuffle,1_{X^*})$ such that
\begin{eqnarray*}
\Li^-_{w}=\Li_{P_w}.
\end{eqnarray*}
Setting $\gamma_{-s_1,\ldots,-s_r}:=\gamma_{\pi_Y(P_w)}$, we get
\begin{eqnarray*}
\gamma_{-s_1,\ldots,-s_r}=\Sum_{k=0}^{s_1+\ldots s_r+r}\frac{a^w_{k}}{k!}.
\end{eqnarray*}

\end{proposition} 
\begin{proof}
The first part of this proposition is the direct consequence of \eqref{gor1} and Proposition \ref{corol5}.
See \cite{inpreparation} for the analytical justification of such algebraic process and the proof of the uniqueness.

\end{proof}

\begin{example}
\begin{eqnarray*}
\gamma_{0}&=&-1 + \frac{1}{1!}\cr
&=&0,\cr
\gamma_{-1}&=&-\frac{1}{1!} +\frac{1}{2!}\cr
&=&-\frac{1}{2},\cr
\gamma_{-2}&=&\frac{1}{1!} -\frac{3}{2!} + \frac{2}{3!}\cr
&=&- \frac{1}{6},\cr
 \gamma_{-3}&=&- \frac{1}{1!} +\frac{7}{2!} - \frac{12}{3!} +\frac{6}{4!}\cr
&=&\frac{3}{4},\cr
\gamma_{-4}&=&\frac{1}{1!} - \frac{15}{2!} + \frac{50}{3!} - \frac{60}{4!} + \frac{24}{5!}\cr
&=&-\frac{7}{15},\cr
\gamma_{0,0}&=&1- \frac{2}{1!} + \frac{1}{2!}\cr
&=&-\frac{1}{2},\cr
\gamma_{0,-1}&=&\frac{1}{1!} -\frac{2}{1!} +\frac{1}{3!}\cr
&=&- \frac{5}{6},\cr
\gamma_{-2,0}&=&- \frac{2}{1!} +\frac{10}{2!} -\frac{14}{3!} + \frac{6}{4!}\cr
&=&\frac{11}{12},\cr
\gamma_{0,-2}&=& -1 +\frac{4}{2!} - \frac{5}{3!}  + \frac{2}{4!}\cr
&=&\frac{1}{4},\cr
\gamma_{-1,-1}&=&-1+ \frac{5}{2!} -\frac{7}{3!} + \frac{3}{4!}\cr
&=&\frac{11}{24},\cr
\gamma_{0,0,0}&=&-1 +  \frac{3}{1!} - \frac{3}{2!} + \frac{1}{3!}\cr
&=&\frac{2}{3},\cr
\gamma_{0,-1,0}&=&-\frac{2}{1!} + \frac{6}{2!} - \frac{6}{3!} + \frac{2}{4!}\cr
&=&\frac{1}{12},\cr
\gamma_{0,-2,0}&=&\frac{-12}{2!}+ \frac{6}{5!} + \frac{2}{1!} -\frac{20}{4!}+ \frac{24}{3!}\cr
&=&-\frac{47}{60},\cr
\gamma_{-4,-4,-6}&=&1 -\frac{16655}{2!} +\frac{5260444}{3!} -\frac{321370622}{4!}+\frac{7519806977}{5!}\cr
&-&\frac{90280292235}{6!}+\frac{647428882810}{7!}-\frac{3028468246320}{8!}\cr
&+&\frac{9748178974760}{9!} -\frac{22298261594760}{10!}+\frac{36869237126640}{11!}\cr
&-&\frac{44258208343200}{12!}+\frac{38240776382400}{13!}-\frac{23192869190400}{14!}\cr
&+&\frac{9376213916160}{15!}-\frac{2270032128000}{16!}+  \frac{ 249080832000}{17!}\cr
&=&-\frac{47315637837661}{137837700}.
\end{eqnarray*}
\end{example}

\section{Conclusion}
\quad In this work, we explained the whole project of extending $\Li_{\bullet}$
over a shuffle subalgebra of rational power series.

In particular, we have studied different aspects of
$\calC\{\Li_w\}_{w\in X^*}$, where $\calC$ denotes the ring
of polynomials in $z,z^{-1}$ and $(1-z)^{-1}$, with coefficents in $\C$.

On the other hand, we applied this new indexing of $\Li_{\bullet}$ to express the polylogarithms
(resp. harmonic sums) at negative multi-indices as polynomials
in $(1-z)^{-1}$ (resp. $N$), with coefficients in $\Z$ (resp. $\Q$). 

We concentrated, particularly, on algebraic and analytic aspects
of this extension allowing index polylogarithms,
at non positive multi-indices, by rational series
and to regularize divergent polyzetas, at non positive multi-indices.
%%%%%%%%%%%%%%%%%%%%%%%%%%%%%%%%%%%%%%%%
%% BIBLIOGRAPHY
%%%%%%%%%%%%%%%%%%%%%%%%%%%%%%%%%%%%%%%%
%\bibliographystyle{alphabet}
%\bibliography{sample}

\end{document}

%% file: macrosB1.tex
\usepackage{amsmath}
\usepackage{amsfonts}
\usepackage{amssymb}
\newcommand{\calA}{{\mathcal A}}

\newcommand{\calH}{{\mathcal H}}
\newcommand{\calC}{{\mathcal C}}

\newcommand{\calS}{{\mathcal S}}

\newcommand{\N}{{\mathbb N}}

\newcommand{\Z}{{\mathbb Z}}
\newcommand{\Q}{{\mathbb Q}}

\newcommand{\C}{{\mathbb C}}

\newcommand{\Frac}[2]{\displaystyle \frac{#1}{#2}}
\newcommand{\Sum}[2]{\displaystyle{\sum_{#1}^{#2}}}

\def\Lyn{{\mathcal Lyn}}

\def\CX{\C \langle X \rangle}
\def\CY{\C \langle Y \rangle}
\def\abs#1{|#1|}

 \def\shuffle{\mathop{_{^{\sqcup\!\sqcup}}}} 

\gdef\stuffle{\;% 
  \setlength{\unitlength}{0.0125cm}% 
  \begin{picture}(20,10)(220,580) 
  \thinlines 
  \put(220,592){\line( 0,-1){ 10}} 
  \put(220,582){\line( 1, 0){ 20}} 
  \put(240,582){\line( 0, 1){ 10}} 
  \put(230,592){\line( 0,-1){ 10}} 
  \put(225,587){\line( 1, 0){ 10}} 
  \end{picture}\; 
}

\newcommand{\Li}{\operatorname{Li}}

\def\L{\mathrm{L}}
\def\H{\mathrm{H}}

\def\F{\mathrm{F}}

\def\deg{\mathrm{deg}}

\newcommand{\poly}[2]{#1 \langle #2 \rangle}

\def\CX{\poly{\C}{X}}

\def\QY{\poly{\Q}{Y}}

\def\CY{\poly{\C}{Y}}

\newcommand{\serie}[2]{#1 \langle \! \langle #2 \rangle \! \rangle}

\def\QY_0{\Q\left\langle{Y_0}\right\rangle}

\def\CX{\C \langle X \rangle}

 \newcommand{\calB}{{\cal B}}

\def\Sum{\displaystyle\sum}

\def\Frac{\displaystyle\frac}
\def\path{\rightsquigarrow}

\def\bv{\mid}

\def\abs#1{\bv\!#1\!\bv}

\def\deg{\mathop\mathrm{deg}\nolimits}

\def\dbinom#1#2{{#1\choose#2}}
\def\binom#1#2{{#1\choose#2}}

%% file: newpack3.tex
\usepackage{color}
\usepackage{ulem}
\def\tcr#1{\textcolor{red}{#1}}

\def\scal#1#2{\langle #1\bv#2 \rangle}
\def\ncp#1#2{#1\langle #2\rangle}
\def\ncs#1#2{#1\langle \!\langle #2\rangle \!\rangle}

%% file: newpack5.tex
\usepackage{color}
\usepackage{ulem}
\def\tcr#1{\textcolor{red}{#1}}

\def\scal#1#2{\langle #1\bv#2 \rangle}
\def\ncp#1#2{#1\langle #2\rangle}
\def\ncs#1#2{#1\langle \!\langle #2\rangle \!\rangle}